\documentclass[%
 reprint,
superscriptaddress,
 amsmath,amssymb,
 aps,
 prx,
floatfix,
]{revtex4-2}

\usepackage{xcolor}
\usepackage{graphicx}
\usepackage{dcolumn}
\usepackage{bm}
\usepackage{mathtools}
\usepackage{braket}
\usepackage[normalem]{ulem}
\usepackage{subfigure}
\DeclarePairedDelimiter{\ceil}{\lceil}{\rceil}

\usepackage{appendix}
\usepackage{amsthm}
\newtheorem{theorem}{Theorem}

\begin{document}

\preprint{APS/123-QED}

\title{Dynamic, Symmetry-Preserving, and Hardware-Adaptable Circuits for Quantum Computing Many-Body States and Correlators of the Anderson Impurity Model}

\author{Eric B. Jones}
        \email{eric.jones@infleqtion.com}
        \affiliation{Infleqtion, Boulder CO, 80301}
\author{Cody James Winkleblack}
        \email{James.Winkleblack@nrel.gov}
        \affiliation{National Renewable Energy Laboratory, Golden CO, 80401}
\author{Colin Campbell}
        \affiliation{Infleqtion, Boulder CO, 80301}
\author{Caleb Rotello}
        \affiliation{National Renewable Energy Laboratory, Golden CO, 80401}
\author{Edward D. Dahl}
        \affiliation{IonQ, College Park, MD 20740}
\author{Matthew Reynolds}
        \affiliation{National Renewable Energy Laboratory, Golden CO, 80401}
\author{Peter Graf}
        \affiliation{National Renewable Energy Laboratory, Golden CO, 80401}
\author{Wesley Jones}
        \affiliation{National Renewable Energy Laboratory, Golden CO, 80401}

\date{\today}

\begin{abstract}
We present a hardware-reconfigurable ansatz on $N_q$-qubits for the variational preparation of many-body states of the Anderson impurity model (AIM) with $N_{\text{imp}}+N_{\text{bath}}=N_q/2$ sites, which conserves total charge and spin z-component within each variational search subspace. The many-body ground state of the AIM is determined as the minimum over all minima of $O(N_q^2)$ distinct charge-spin sectors. Hamiltonian expectation values are shown to require $\omega(N_q) < N_{\text{meas.}} \leq O(N_{\text{imp}}N_{\text{bath}})$ symmetry-preserving, parallelizable measurement circuits, each amenable to post-selection. To obtain the one-particle impurity Green's function we show how initial Krylov vectors can be computed via mid-circuit measurement and how Lanczos iterations can be computed using the symmetry-preserving ansatz. For a single-impurity Anderson model with a number of bath sites increasing from one to six, we show using numerical emulation that the ease of variational ground-state preparation is suggestive of linear scaling in circuit depth and sub-quartic scaling in optimizer complexity. We therefore expect that, combined with time-dependent methods for Green's function computation, our ansatz provides a useful tool to account for electronic correlations on early fault-tolerant processors. Finally, with a view towards computing real materials properties of interest like magnetic susceptibilities and electron-hole propagators, we provide a straightforward method to compute many-body, time-dependent correlation functions using a combination of time evolution, mid-circuit measurement-conditioned operations, and the Hadamard test.

\end{abstract}

\maketitle


\section{\label{sec:intro}Introduction}

The electronic structure problem---solving for the ground states of collections of interacting electrons---is famously difficult. In its full generality, the electronic structure problem resides in the Quantum Merlin Arthur (QMA) complexity class, which is comprised of problems whose solutions are thought to be difficult to find even on a quantum computer \cite{schuch2009computational, PRXQuantum.3.020322}. In spite of this, a number of classical approximation techniques, like density functional theory (DFT) and GW theory, are known to describe many weakly- and moderately-correlated electronic systems well \cite{RevModPhys.87.897, PhysRevLett.96.226402}. Meanwhile, it remains unclear what degree of advantage quantum computing alone will confer over classical methods for the electronic structure problem \cite{lee2023evaluating}. Moreover, quantum processors remain relatively small and noisy, with the largest processors ranging from dozens to hundreds of qubits with typical two-qubit gate infidelities on the orders of .1\% to 1\% \cite{kim2023evidence, moses2023race, morvan2023phase}. As such, a compelling near-term avenue to the demonstration of quantum utility in the electronic structure of real materials and chemical systems is to marry the scalability of weakly-correlated theories, like DFT and GW, with theories that can treat strong electronic correlations on a subsystem self-consistently via embedding theory \cite{tomczak2015qsgw+, choi2016first, PhysRevB.95.155104, clinton2024towards}.

Dynamical mean-field theory (DMFT) is one such embedding theory, which is particularly well-suited to periodic systems \cite{RevModPhys.68.13, kotliar2006electronic}. In DMFT, one replaces a lattice with interactions on every site \footnote{Generally, a ``site'' can refer to an atom, a collection of atoms, or some collection of interacting electronic orbitals. Here we consider it a single spatially localized site with one spin-up and one spin-down degree of freedom.} with two, simultaneous but complementary descriptions of the dynamics. The first description neglects all interactions on the extended lattice but retains hopping between sites so that the quasi-particle band picture of the electronic dynamics is considered valid. For real materials, this description can be treated with theories like DFT or GW. The second description considers only a specific site, an impurity or cluster, within the lattice and treats all on-site interactions and transitions exactly, but approximates the coupling of the site to the rest of the lattice as electrons hopping into and out of an effective bath. This description is similar to the mean-field treatment of the Ising model, where the local Green's function plays the analogous role of the on-site magnetization and the dynamic hybridization function plays the role of the mean-field \cite{martin2016interacting}. Self-consistency between the non-interacting lattice and impurity descriptions is enforced by insisting that the local lattice Green's function and self-energy are equivalent to the impurity Green's function and self-energy. Self-consistency is obtained by varying the hybridization function, termed bath fitting, and constructs a solution to the original problem, which can simultaneously describe both delocalized, band-like behavior, and localized dynamical correlations.

The most computationally burdensome aspect of DMFT is solving for the impurity Green's function and, by extension, self-energy. Classically, this is often done either via exact diagonalization (ED) \cite{PhysRevLett.72.1545}, which scales exponentially in the number of impurity and bath orbitals---the number of single-electron degrees of freedom comprising the impurity and a discrete representation of the bath, respectively---or via continuous-time quantum Monte Carlo (CTQMC), which scales exponentially only in the number of impurity orbitals, but which suffers from the fermionic sign problem \cite{RevModPhys.83.349, PhysRevB.74.155107, PhysRevB.75.155113}. And while recent work has proposed an algorithm for solving impurity models with quasipolynomial runtime in the number of bath orbitals using fermionic Gaussian states, the method still scales exponentially in the number of impurity orbitals \cite{bravyi2017complexity}. Exact diagonalization has been used to solve impurities of up to about 25 (spatial) orbitals, while more specialized quantum chemical approximation schemes have treated systems of up to 127 orbitals \cite{PhysRevB.73.205121, PhysRevX.11.021006}. Tensor network  and renormalization group methods show the promise to more scalably solve the impurity problem \cite{PhysRevB.90.115124, PhysRevB.92.155132, PhysRevX.5.041032, jamet2023anderson, bulla2008numerical, PhysRevLett.93.246403}, but the community has naturally looked towards quantum computers as alternative impurity solvers.

Methods to solve the impurity problem via quantum circuits target either the frequency-dependent Green's function or the time-dependent Green's function, which must then be transformed to enforce self-consistency. Initial work by Bauer et al. and Kreula et al. proposed measuring Green's functions of the single-impurity Anderson model (SIAM) on the real-time axis via Trotterized time evolution and ancilla qubits \cite{PhysRevX.6.031045, kreula2016few}. Subsequently, Rungger et al. demonstrated the computation of the impurity Green's function directly on the real-frequency axis using the variational quantum eigensolver (VQE) to construct many-body states in the Lehmann representation on quantum devices \cite{rungger2019dynamical}. Keen et al. fit real-time Green's functions derived from Trotterized evolution on a quantum device to extract frequency dependent Green's functions \cite{keen2020quantum}. Lie algebraic fast-forwarding has extended the reach of near-term processors to calculate real-time Green's functions as was recently demonstrated by Steckmann et al. \cite{PhysRevResearch.5.023198}. And the quantum equation of motion method was used to compute Green's functions in the Lehmann representation to perform DMFT on 14 qubits of an IBM quantum computer as well \cite{selisko2024dynamical}.  Still, many of the aforementioned techniques suffer from problems of scalability, either from too-deep circuits required by Trotter error \cite{jaderberg2020minimum}, a combinatorial number of many-body states, which need to be constructed in the Lehmann representation \cite{rungger2019dynamical}, or an exponential growth in the Hamiltonian algebra as a function of bath sites for fast-forwarding \cite{PhysRevResearch.5.023198}.

An alternative proposed recently by Jamet et al. constructs the impurity Green's function on the real or imaginary frequency axis via its continued fraction expression and a variational version of the Lanczos iteration \cite{jamet2021krylov}. While the associated number of many-body states that need to be prepared on a quantum computer also scales combinatorially in the number of total orbitals, the associated Krylov subspace can often be truncated to provide a good approximation to the Green's function before the full Krylov space dimension is reached. However, experimental implementations of the Krylov variational quantum algorithm (KVQA) have been few \cite{PhysRevB.107.165155}. The main challenge in implementing a scalable version of the KVQA is to find an ansatz that does not suffer from barren plateaus in its gradient, as does the low-depth hardware-efficient ansatz (HEA) \cite{mcclean2018barren, cerezo2021cost}, but which also does not suffer from high Trotterization and SWAP routing overhead as do more physically-informed ansatzae, such as quantum alternating operator and unitary coupled cluster ansatzae \cite{kremenetski2021quantum, anand2022quantum}.

\begin{figure}
\centering
\includegraphics[width=.8\linewidth]{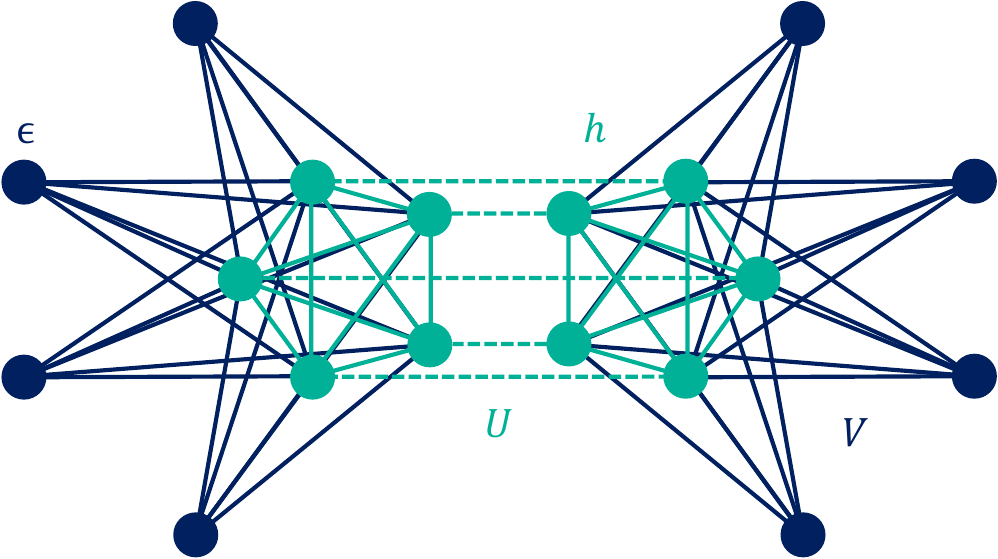}
\caption{\label{fig:AIM} Topology of the charge and spin conserving $\mathbb{Z}_2$-symmetric Anderson impurity model with on-site interactions in the the impurity. There are five impurity sites and four bath sites. The left and right subgraphs demark the two spin registers. Light green solid lines represent hopping between and on-site energies of impurity sites (light green circles) given by the matrix, $h$, and a complete graph topology within each spin register. Light green dashed lines represent on-site interactions between spin-up and spin-down sites in the impurity given by the vector, $U$. Dark purple lines depict hybridization between the impurity and bath sites (dark purple circles) given by the matrix $V$, and forming a complete bipartite graph with the impurity orbitals in each spin register. The diagonal bath energies are contained in the vector $\epsilon$. In this work we restrict the number of impurity sites to one.}
\end{figure}

Motivated by prior work on symmetry-preserving ansatzae (SPA) for VQE \cite{gard2020efficient}, we provide a variational ansatz for the preparation of many-body states of the Anderson impurity model (AIM), which describes the impurity problem in the context of DMFT. Our ansatz has the properties that it is (1) physics-informed and manifestly constrained by the symmetries of the AIM, enabling efficient trainability and high expressibility at low gate depth; and (2) adaptable to arbitrary hardware topologies with little overhead. Our ansatz differs from that in Ref. \cite{gard2020efficient} in that it explicitly accounts for interactions within the impurity. In Sec.~\ref{sec:spa} we introduce the SPA and numerically analyze its expressibility and trainability to prepare ground states of the single-impurity Anderson model with a varying number of bath sites. Then, in Sec.~\ref{sec:spkvqa} we apply our ansatz to the construction of one-particle Green's functions via the variational preparation of a sequence of Lanczos vectors. We remark that the focus of this paper is on assessing the expressivity and trainability of the SPA for a large number of Hamiltonian parameter instantiations rather than closing the DMFT loop. While we expect that our ansatz will find wide utility in the preparation of many-body states of the AIM, we ultimately expect it to find the greatest utility in schema where ground states of the AIM are encoded on a quantum processor either variationally or deterministically via classical tensor network methods and parameter fixing while Green's functions are computed using direct time evolution \cite{jamet2023anderson, provazza2024fast}. For this reason, we propose a new method for computing generic $m$-point correlation functions of fermionic theories in Sec.~\ref{sec:correlators}. We discuss future avenues of work such as this in Sec.~\ref{sec:conclusion}.

\section{\label{sec:spa}Symmetry-Preserving Ansatz}

The Hamiltonian that governs the impurity dynamics in DMFT is the Anderson impurity model (AIM),

\begin{equation}\label{eq:AIM}
\begin{split}
H &= H_{imp} + H_{hyb} + H_{bath}, \\
H_{imp} &= \sum_{i, j=1}^{N_{imp}} \sum_{\sigma=\downarrow}^{\uparrow} h_{ij} c^{\dagger}_{i \sigma} c_{j \sigma} + \sum_{i=1}^{N_{imp}} U_i n_{i \uparrow} n_{i \downarrow}\\
H_{hyb} &= \sum_{i=1}^{N_{imp}} \sum_{b=1}^{N_{bath}} \sum_{\sigma=\downarrow}^{\uparrow} V_{i b} c_{i \sigma}^{\dagger} c_{b \sigma} + h.c.\\
H_{bath} &= \sum_{b=1}^{N_{bath}} \sum_{\sigma=\downarrow}^{\uparrow} \epsilon_b c_{b \sigma}^{\dagger} c_{b \sigma}.
\end{split}
\end{equation}

The AIM in Eq.~\ref{eq:AIM}, and depicted in Fig.\ref{fig:AIM}, is a slight generalization of the SIAM typically used in DMFT calculations and represents the possibility to include more than one impurity site, which can approximate strong spatial correlations otherwise neglected in DMFT \cite{he2014quantum, martin2016interacting}. $H_{imp}$ is comprised of hopping terms ($h_{ij}$) within, on-site energies ($h_{ii}$) of, and interactions ($U_i$) on, $N_{imp}$ impurity sites. $H_{hyb}$ represents hopping, or hybridization, between impurity and bath states with strength $V_{ib}$. And $H_{bath}$ forms a discrete representation of the bath with $N_{bath}$ on-site energies $\epsilon_b$.

\begin{figure}
\centering
\includegraphics[width=\linewidth]{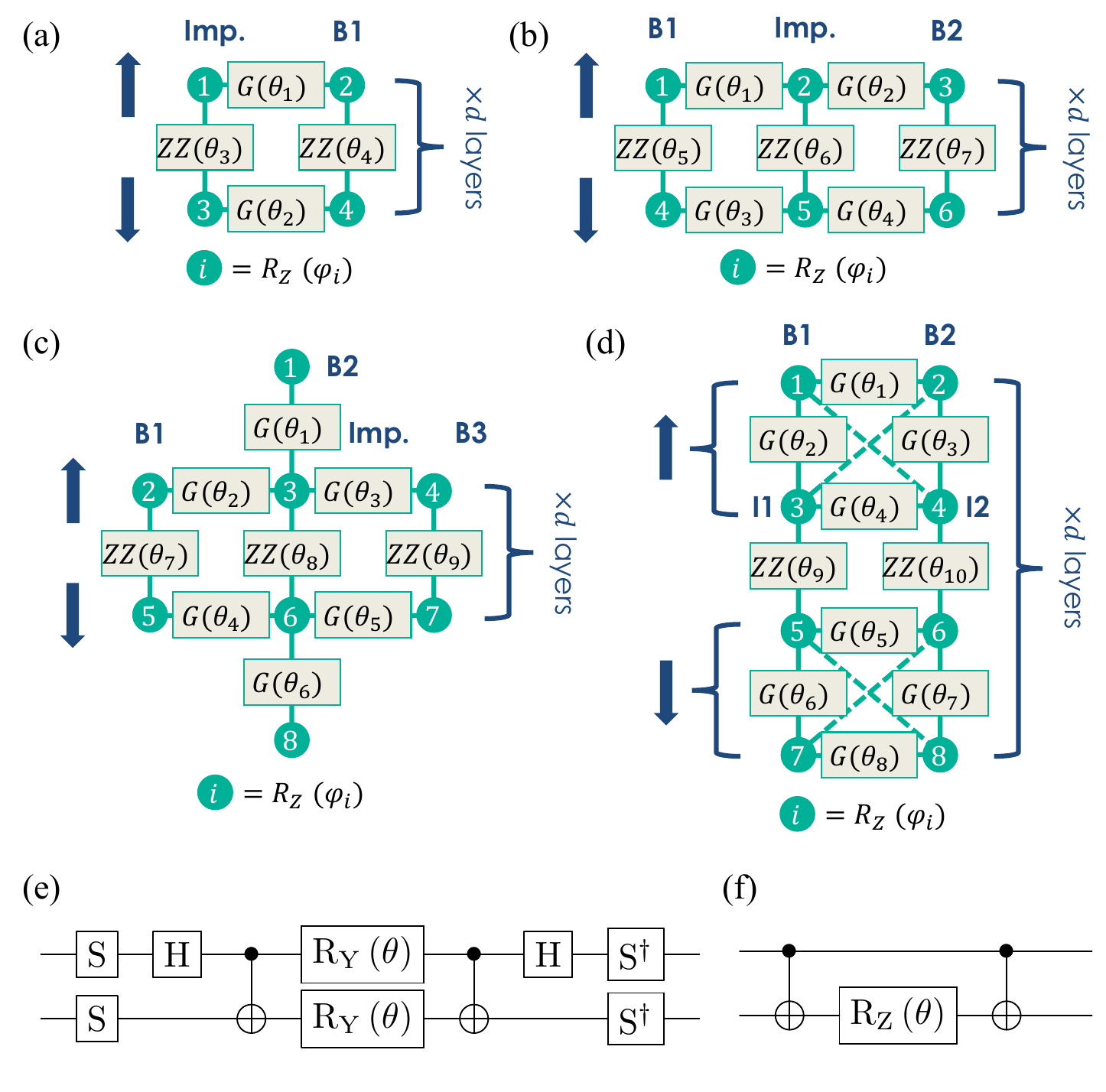}
\caption{\label{fig:ansatz} Symmetry-preserving ansatzae in a two-dimensional, square lattice architecture, $U_{SPA}(\boldsymbol{\theta})$, for (a) one-impurity, one-bath, (b) one-impurity, two-bath, (c) one-impurity, three bath, and (d) two-impurity, two-bath site Anderson impurity models. Solid green circles (lines) represent qubits (planar interactions) with associated $Rz(\varphi_i)$ (Givens and $ZZ$-rotation) gates. Dashed lines represent long-range interactions. (e) Compilation of Givens rotation into a standard gate-set. (f) Compilation of $ZZ$-rotation into a standard gate-set.}
\end{figure}

As is well-known, Eq.~\ref{eq:AIM} commutes with the total charge, $n^+$, and spin z-component, $n^-$, operators, 

\begin{equation}\label{eq:symmetries}
n^{\pm} = \sum_{i} (n_{i \uparrow} \pm n_{i \downarrow}) + \sum_{b} (n_{b \uparrow} \pm n_{b \downarrow}),
\end{equation}
each of which generates a symmetry of the AIM that we thus refer to a charge and spin symmetry, respectively. In addition, Eq.~\ref{eq:AIM} has an overall up-down spin degeneracy---a $\mathbb{Z}_2$ symmetry, which can be broken, for example, by the presence of an external magnetic field. Each of these symmetries is manifestly preserved under the Jordan-Wigner (J-W) mapping from fermionic operators to qubit operators, $c^{\dagger}_{\mu} = \frac{1}{2} (X_{\mu} - i Y_{\mu}) \prod_{\nu < \mu} Z_{\nu}$, where we organize qubit indices such that $\mu \in \{1, \ldots, N_{imp} + N_{bath}\}$ correspond to spin up qubits and $\mu \in \{N_{imp} + N_{bath}+1, \ldots, 2N_{imp} + 2N_{bath} = N_{q}\}$ correspond to spin down qubits. Due to the commutativity between Eq.~\ref{eq:AIM} and the operators in Eq.~\ref{eq:symmetries}, the AIM can be diagonalized into simultaneous eigenstates of spin, charge and energy, indicating that we can partition a variational search over the full $N_{q}$-qubit Hilbert space for the ground state of Eq.~\ref{eq:AIM} into a collection of searches over smaller subspaces. There are $N_q - 1$ non-trivial charge sectors and within each charge sector, $N\in \{1,\ldots, N_q-1\}$, there are $\min(N,N_q-N) + 1$ spin sectors, of which only $\ceil{(\min(N,N_q-N) + 1)/2}$ are unique due to the $\mathbb{Z}_2$ symmetry of Eq.~\ref{eq:AIM}. Hence, if one does not know, a priori, which charge-spin sector the ground state of a specified AIM resides in, one need only perform $O(N_q^2)$ separate runs of VQE to determine the global ground state. If instead, one has a reliable reference state, one can run a symmetry-preserving VQE routine within a single charge-spin sector. Within each charge-spin sector the dimension of the symmetry-constrained subspace is $D(N,S_z)={N_q/2 \choose N_{\uparrow}} {N_q/2 \choose N_{\downarrow}} = {N_q/2 \choose (N+S_z)/2} {N_q/2 \choose (N-S_z)/2}$. It can be verified that $\sum_N \sum_{S_z(N)} D(N,S_z) = 2^{N_q}$ as expected.

\begin{figure}
\centering
\includegraphics[width=\linewidth]{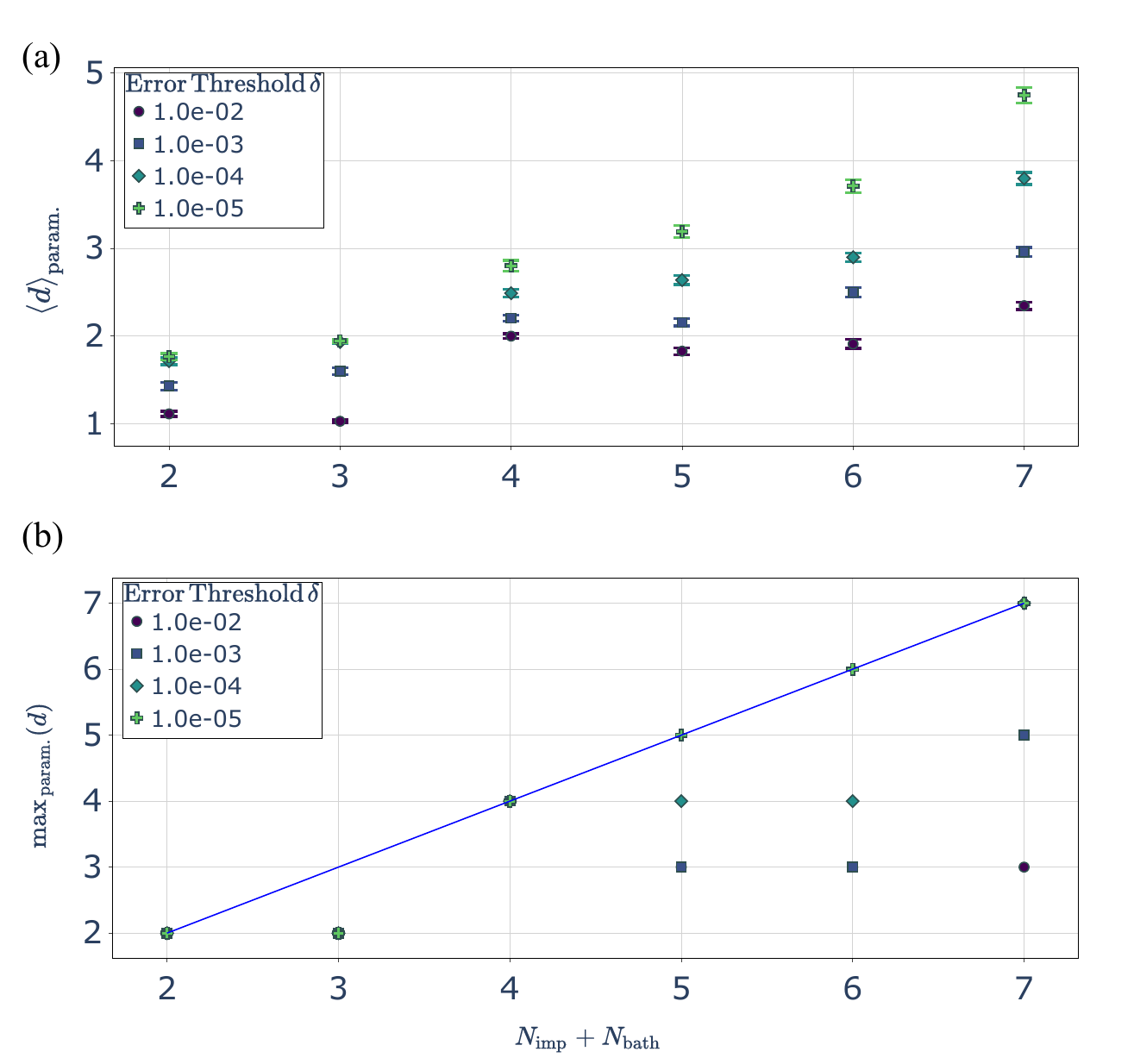}
\caption{\label{fig:vqe_vs_sites} Expressibility of the symmetry-preserving ansatz on a two-dimensional planar qubit topology for the single impurity Anderson model. (a) Hamiltonian seed-averaged depth, $\langle d \rangle_{\text{param.}}$, at which the ansatz is able to approximate the exact ground state as a function of the number of sites, $N_{\text{imp}}+N_{\text{bath}}$ for various ground state overlap error thresholds. Error bars represent one standard error of the mean over the 150 Hamiltonian seeds. (b) Scaling of the ansatz depth, $\max_{\text{param.}}(d)$, required to approximate the exact ground state for the worst-case Hamiltonian seed at each system size. The blue line indicates the line $\max_{\text{param.}}(d)=N_{\text{imp}}+N_{\text{bath}}$ for the most stringent error threshold $\delta = 10^{-5}$.}
\end{figure}

After performing the J-W transformation, an $(N, S_z)$-sector can be instantiated at the start of a circuit by performing $N_{\uparrow} = (N+S_z)/2$ bit flips on the all-zeros, $\ket{0}^{\bigotimes N_q/2}_{\uparrow}$, spin-up qubit register and $N_{\downarrow} = (N-S_z)/2$ bit flips on the all-zeros, $\ket{0}^{\bigotimes N_q/2}_{\downarrow}$, spin-down qubit register. Following Ref.~\cite{gard2020efficient}, we place the corresponding Pauli-$X$ gates as evenly-spaced as possible within each register. Then, the interpretation of combined charge-spin symmetry in any subsequent VQE ansatz is that operations should preserve total Hamming weight within each qubit register. Note that these initial excitations should technically also include Pauli-$Z$ gates from the J-W transformation, but their effect is only to generate a global phase on a classical product state, which we can ignore. Fig.~\ref{fig:ansatz} shows the form of the SPA for various system sizes. In addition to symmetry considerations, the structural instantiation of the SPA is also informed by the architecture of the processor on which it is run. For instance, it can be easily adapted to a two-dimensional, (quasi-)planar qubit connectivity such as some neutral atom \cite{graham2022multi} and superconducting \cite{arute2019quantum} platforms.

Fig.~\ref{fig:ansatz}a depicts the SPA for a two-site AIM where from top left to bottom right the qubits are ordered as $(\text{Imp.}, \uparrow), \: (\text{Bath}, \uparrow), \:(\text{Imp.}, \downarrow), \: (\text{Bath}, \downarrow)$. Mixing of different charge configurations within each spin register is accomplished by parameterized Givens rotations, $G(\theta_i)$ (note that fractional $\text{iSWAP}(\theta_i)$ gates could also be used here). Meanwhile, Hamming weight preservation within each spin register is accomplished by coupling the spin up and down registers only via two-qubit phase rotations, $ZZ(\theta_i)=e^{-i \theta_i ZZ/2}$, motivated by the Hubbard $U$ interaction in Eq.~\ref{eq:AIM}. Finally, an independent $R_Z(\varphi_i)$ rotation is placed on each qubit at the end of the layer to represent on-site energies and cancel unwanted phases. This single-layer structure can then be repeated $d$ times. Figs.~\ref{fig:ansatz}b-c demonstrate the generalization of the ansatz for AIMs with two and three bath sites, respectively.

The SPA shares similarities with the Hamiltonian Variational Ansatz (HVA) for the AIM \cite{libbi2022effective}. In particular, the hardware embedding of the SPA attempts to exploit the sparsity structure of Eq.~\ref{eq:AIM} wherein electrons can undergo impurity-impurity and impurity-bath hopping, but not bath-bath hopping, as the bath is generally assumed to be diagonal. On the other hand, the SPA will leverage hardware connections, such as $ZZ(\theta_i)$ gates between the spin up and spin down qubits of bath sites B1 and B2 in Fig.~\ref{fig:ansatz}b, which would not be present in the HVA, as it would only couple spin up and down qubits on the impurity. The effect of these inclusions is to increase expressivity while maintaining symmetry constraints.

\begin{figure*}[]
    \centering
    \includegraphics[width=\linewidth]{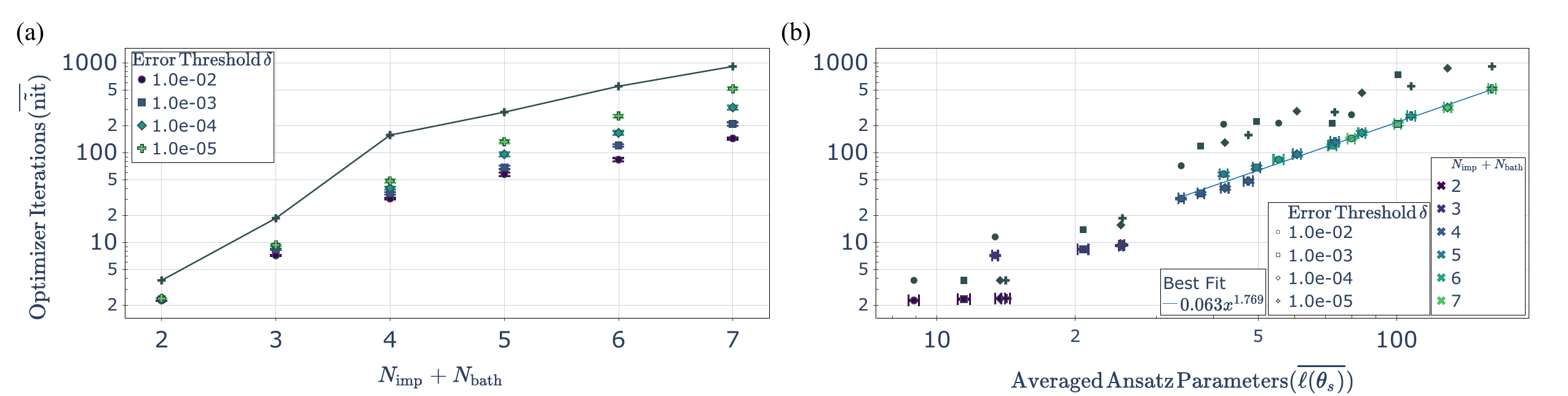}
    \caption{Ground state trainability of the symmetry-preserving ansatz on a two-dimensional planar qubit topology for the single impurity Anderson model. (a) On a semi-log plot, $\overline{\Tilde{\text{nit}}}$ is the Hamiltonian seed-averaged number of optimizer iterations, totaled over all charge-spin sectors, required to converge to a target overlap error, $\delta$, for a particular site number, $N_{\text{imp}}+N_{\text{bath}}$, normalized by the number of charge-spin sectors searched over for the variational minimization, which is $2 \times (N_{\text{imp}}+N_{\text{bath}} + 1)^2$. Error bars represent standard error of the mean over Hamiltonian seeds. Solid black curve is a visual cue for worst-case scaling at the $10^{-5}$ error threshold. (b) $\overline{\Tilde{\text{nit}}}$ as a function of $\overline{\ell(\pmb{\theta}_s)}$, which is the ``Averaged Ansatz Parameters''-- the seed-averaged length of each $\bm{\theta}_s$ vector, where the length is given by $d \times (N_e + N_q)$, $N_e$ being the number of edges in the SPA, and where the average is taken over a Hamltonian parameter ensemble labeled by a target error threshold and a site number, ($N_{\text{bath}} + N_{\text{imp}}$). Blue line is a power-law fit and the black markers above the main fit are worst-case instances.} \label{fig:optimizer_iterations}
\end{figure*}

Fig.~\ref{fig:ansatz}d demonstrates how the SPA can be expanded to include more than one impurity site, as well as to utilize longer-range gate interactions, in this case represented by the next-nearest neighbor (NNN) green dashed lines. Here again, the SPA corresponds closely to the HVA with explicit hopping terms representable between both impurity and both bath orbitals. NNN interactions could also generate $ZZ(\theta_i)$ interaction gates between the two impurity orbitals for AIMs with not stricly site-local electronic interactions. Generalizing further, six bath orbitals could be coupled to each of the impurity orbitals in Fig.~\ref{fig:ansatz} given NNNN connectivity (distance-2 plus shapes). Quasiplanar neutral atom platforms are particularly well-suited to such generalizations, with the Rydberg interaction able to naturally generate long-range interactions \cite{baker2021exploiting, RevModPhys.82.2313}.

Let $\ket{\boldsymbol{0}} \equiv \ket{0}^{\bigotimes N_q}$. We denote a variational trial state prepared by the SPA as

\begin{equation}\label{eq:ns_trial}
\begin{split}
\ket{N,S_z;\boldsymbol{\theta}} &= U_{(N, S_z)}(\boldsymbol{\theta}) \ket{\boldsymbol{0}}, \\ U_{(N, S_z)}(\boldsymbol{\theta}) &= U_{SPA}(\boldsymbol{\theta}) \prod_{\mu \in init. ex.} X_{\mu},
\end{split}
\end{equation}
where $\boldsymbol{\theta} = (\theta_1, \ldots, \varphi_1, \ldots, \varphi_{N_q}, \ldots, \varphi_{d\times N_q})$, $U_{SPA}(\boldsymbol{\theta})$ is a circuit of the form defined in Fig.~\ref{fig:ansatz} (implicitly defined to some depth, $d$), and $U_{(N, S_z)}(\boldsymbol{\theta})$ includes the initial bit flips to initialize the ansatz in the correct charge-spin sector. An approximation for the global ground state is then determined via minimizing the symmetry-constrained Hamiltonian minimization over all charge-spin sectors:

\begin{equation}\label{eq:ground_state}
\ket{\widetilde{GS}} = \min_{N, S_z} \Big[ \min_{\boldsymbol{\theta}} \braket{N, S_z; \boldsymbol{\theta}|H|N, S_z; \boldsymbol{\theta}} \Big].
\end{equation}
In order to perform the minimization in Eq.~\ref{eq:ground_state} in a fully symmetry-preserving fashion, one needs to be able to compute Hamiltonian expectation values from measurements, i.e., operator strings, that also commute with Eq.~\ref{eq:symmetries}. Fortunately, the structure of Eq.~\ref{eq:AIM} also facilitates this. Ignoring terms that are proportional to the identity, there are only two types of terms in Eq.~\ref{eq:AIM} after performing the J-W transformation. The first type of term is that which can be computed from reading out all qubits in the computational $Z$-basis. These are the terms that come from the Hubbard-$U$ interaction(s) and diagonal terms of $h_{ij}$ in $H_{imp}$ as well as every term in $H_{bath}$. Any string of Pauli-$Z$ operators commutes with Eq.~\ref{eq:symmetries}. The other type of terms are the hopping terms, which appear in $H_{hyb}$ and include the off-diagonal elements of $h_{ij}$ in $H_{imp}$. A hopping term has the general form $\propto (X_{\mu} X_{\nu} + Y_{\mu} Y_{\nu}) \prod_{\mu < \rho < \nu} Z_{\rho}$, which means its expectation value can be evaluated by performing a single Givens rotation, $G(-\pi/4)$, on the qubit pair $(\mu, \nu)$, followed by readout of all qubits in parallel. This readout scheme corresponds to measuring the operator $ (1/2)(X_{\mu} X_{\nu} + Y_{\mu} Y_{\nu}) \prod_{\rho \neq \mu, \nu} Z_{\rho}$, which also commutes with Eq.~\ref{eq:symmetries} and whose eigenvectors thus conserve charge and spin on the whole qubit array. In fact, multiple hopping terms can be measured in parallel in this manner as long as the qubits that the associated Givens rotations affect do not overlap. Measuring the expectation values of Eq.~\ref{eq:AIM} in this way, with only $Z$-basis and hopping circuits, ensures symmetries are preserved and as an added benefit enables the use of symmetry-based post-selection as an error mitigation technique \cite{arute2020observation, jones2022small, PhysRevResearch.5.033082}. In the absence of measurement circuit parallelization, the number of circuits required to measure the expectation value of Eq.~\ref{eq:AIM} is $1+2N_{imp}N_{bath}+N_{imp}(N_{imp}-1)$. In the presence of measurement circuit parallelization, the spin-up and spin-down terms can be measured in parallel, reducing each of the latter two terms by a factor of two. Either way, an upper bound of $O(N_{imp}N_{bath})$ circuits is required to compute the Hamiltonian expectation value, assuming $N_{\text{bath}} > N_{\text{imp}}$ as is nearly always the case in DMFT. Furthermore, if the parity terms in each measurement circuit are ignored, parallelized measurement can first be performed on the hopping terms between the impurity and bath orbitals, each hopping term corresponding to an edge in a complete, bipartite graph. The chromatic index of such a graph is $\max(N_{\text{imp}},N_{\text{bath}})$-- the maximum degree of any node therein. Meanwhile, each hopping term between two impurity orbitals constitutes an edge in a complete graph. The chromatic index of a complete graph is either $N_{\text{imp}}$ or $N_{\text{imp}}-1$ depending on if $N_{\text{imp}}$ is odd or even, respectively. Hence, a loose lower bound on the number of parallelized measurement circuits that need to be run to compute the expectation value of Eq.~\ref{eq:AIM} scales as $1+N_{\text{bath}}+N_{\text{imp}}+ (-1) \sim \omega(N_q)$. However, we note that the time complexity of measuring these observables can be converted to space complexity by using $\omega(N_q) < N_{\text{meas.}} \leq O(N_{\text{imp}}N_{\text{bath}})$ measurement circuits run trivially in parallel on $N_{\text{meas.}}$ smaller processors. In future work, it would be interesting to investigate the applicability of recent number conserving or fermionic shadow tomography techniques to measuring expectation values \cite{hearth2023efficient, king2024triply}.

Fig.~\ref{fig:vqe_vs_sites} summarizes the finite-size expressivity of the SPA, Eq.~\ref{eq:ns_trial}, in its ability to solve for the ground state of Eq.~\ref{eq:AIM} via the minimization condition in Eq.~\ref{eq:ground_state}. For AIMs with $N_{\text{imp}}+N_{\text{bath}} \in \{2, 3, 4, 5, 6, 7\}$ (corresponding to $N_q \in \{4, 6, 8, 10, 12, 14 \}$) with a single impurity, we draw 150 AIM Hamiltonians from uniform parameter distributions $h_{ij} \in [-5, 5]$, $U_i \in [1, 10]$, $V_{ib} \in [-5, 5]$, and $\epsilon_b \in [-5, 5]$ (all energies being in electronvolts, eV). These distributions reflect a range of impurity parameters typically associated with performing DMFT on real materials. We then compute the ansatz depth, $d$, required to converge the overlap of the variational ground state estimate with the true ground state to below some target error, $\delta = 1 - |\braket{\widetilde{GS}|GS}|$. $|\braket{\widetilde{GS}|GS}|$ is the Uhlmann fidelity, which has the interpretation as a distance metric between two rays in Hilbert space \cite{manenti2023quantum}. Fig.~\ref{fig:vqe_vs_sites} shows the scaling of the Hamiltonian parameter-averaged depth, $\langle d \rangle_{\text{param.}}$, as a function of the number of AIM sites, $N_{\text{imp}}+N_{\text{bath}}$, for various target ground state errors $\delta \in \{10^{-2}, 10^{-3}, 10^{-4}, 10^{-5} \}$. Hamiltonian construction and exact diagonalization were performed using \textit{OpenFermion} \cite{mcclean2020openfermion}, and all quantum circuits, compiled to common elementary gates as shown in Fig.~\ref{fig:ansatz}e-f, were simulated in the absence of shot or gate noise using the \textit{Qulacs} circuit simulation framework \cite{suzuki2021qulacs}. To perform the classical minimiztion within each charge-spin sector we used the \textit{SciPy} implementation of the Broyden–Fletcher–Goldfarb–Shanno (BFGS) algorithm \cite{virtanen2020scipy}. Of the 150 sets of Hamiltonian parameters seeded for each system size, $\{ 123, 135, 142, 140, 146, 144 \}$ seeds were used, respectively, to compute data points and error bars, which represent standard error of the mean. All discarded seeds, other than three seeds at seven sites discarded due to the failure of the optimizer to converge after hitting an emulation wall-time, involved a degeneracy in the ground state calculated by exact diagonalization. In principle, our ground state preparation method can be extended to ground states with a degeneracy by creating superpositions of charge-spin sector ground states with equivalent energies, but for the sake of simplicity, we discard these classical results.

While it is difficult to make conclusions regarding asymptotic performance based on finite-size simulations, Fig.~\ref{fig:vqe_vs_sites} suggests that the SPA has the generic capacity to prepare ground states on $N_q$ qubits in low-order polynomial depth---roughly linear for the most stringent target ground state overlap error. Hence it can be regarded as efficient in its expressibility. For the purposes of variational state preparation, however, trainability of an ansatz is also important. To assess trainability of the SPA, we define a quantity, $\overline{\Tilde{\text{nit}}}$, which is the Hamiltonian seed-averaged number of BFGS optimizer iterations, totaled over all charge-spin sectors, required to converge to a target overlap error, $\delta$, for a particular site number, $N_{\text{imp}}+N_{\text{bath}}$, normalized by the number of charge-spin sectors searched over for the variational minimization, which is $2 \times (N_{\text{imp}}+N_{\text{bath}} + 1)^2$. Because we have already accounted for the quadratic overhead in searching through the different charge-spin sectors, we are interested in the ability of a classical optimizer to determine the ground state in some typical sense within any charge-spin sector (averaged over many AIM Hamiltonian realizations). This is what the scaling of $\overline{\Tilde{\text{nit}}}$ addresses as a function of system size in Fig.~\ref{fig:optimizer_iterations}a, which shows sub-exponential scaling for all target error thresholds on a semi-log plot. To more precisely assess the scaling of $\overline{\Tilde{\text{nit}}}$ as a function of ansatz depth, we define another quantity, $\overline{\ell(\pmb{\theta}_s)}$, which is the ``Averaged Ansatz Parameters''---the seed-averaged length of each $\bm{\theta}_s$ vector, where the length is given by $d \times (N_e + N_q)$, $N_e$ being the number of edges in the SPA, and where the average is taken over a Hamltonian parameter ensemble labeled by a target error threshold and a site number, ($N_{bath} + N_{imp}$). Fig.~\ref{fig:optimizer_iterations}b shows how $\overline{\Tilde{\text{nit}}}$ scales with $\overline{\ell(\pmb{\theta}_s)}$ on a log-log plot. Remarkably, the number of optimizer calls required to find the ground state as a function of the number of ansatz parameters scales roughly as a power law on average with fitted exponent $\sim 1.769$.

Taken together, Figs.~\ref{fig:vqe_vs_sites}-\ref{fig:optimizer_iterations} indicate that the SPA is efficient in both its expressibility and trainability in preparing ground states of the single-impurity Anderson model as the discretization of the bath becomes more fine-grained. Namely, Fig.~\ref{fig:vqe_vs_sites}a indicates that the average-case depth to reach the ground state is no worse than roughly linearly-scaling. In Fig.~\ref{fig:vqe_vs_sites}b we show that for the most stringent error threshold considered, $\delta=10^{-5}$, the worst-case scaling in depth is exactly linear after (or excluding) $N_{\text{imp}}+N_{\text{bath}}=3$, which is consistent with near-saturation of the Lieb-Robinson bound \cite{lieb1972finite, bravyi2006lieb}. Moreover, given the numerical evidence in Fig.~\ref{fig:optimizer_iterations}, we conjecture that the number of optimizer calls scales, in the average case, as not much worse than $\sim N_q^{3.538}$ in the number of sites, since the number of ansatz parameters scales quadratically in the number of sites (one factor per layer and linear depth to represent the ground state). Note also that worst-case scaling over the parameter set in Fig.~\ref{fig:optimizer_iterations} scales similarly---being shifted up by a constant factor. Hence, if one is provided a reliable reference state---that is, a known charge-spin sector within which to search for the ground state---as might be determined by Hartree-Fock or a higher-fidelity mean-field theory \cite{jiang2018quantum, sun2020recent, PhysRevLett.96.226402}, $O(1)$ fixed space-parallelized measurements, and parallelized gate operation within each ansatz layer (as is usually assumed), then the finite-size scaling observed is consistent with $\Theta(N_q^{4.538})$ runtime to prepare the ground state of the AIM. Interestingly, this heuristic, average-case scaling is somewhat worse than Bravyi and Gossett's result, which established an algorithm for estimating the ground state energy of quantum impurity models with runtime $O(N_{\text{bath}}^3)\exp[O(N_{\text{imp}}\log^3(N_{\text{imp}}\gamma^{-1}))]$ \cite{bravyi2017complexity}. For $N_{\text{imp}}=1$, like is used in this work, this scaling reduces to $O(N_q^3)$ in the number of sites (and is quasi-linear in $\gamma^{-1}$). However, ground-state preparation is only one of the core subroutines involved in performing the impurity component of DMFT computations, the other being one-particle Green's function construction. Without an efficient method to compute Green's functions classically, an efficient classical ground state preparation algorithm alone is insufficient to perform DMFT efficiently. Hence, we now discuss an application of our ansatz to the problem of computing Green's functions via the continued fraction representation in frequency space. Subsequently, we provide an efficient method for computing $m$-point correlation functions directly in the time domain.
 
\section{\label{sec:spkvqa}Symmetry-Preserving Krylov Variational Quantum Algorithm}

In DMFT, once a ground state of the AIM has been prepared, the single-particle Green's function must be computed in order to enforce self-consistency of the impurity self-energy with the local bath self-energy \cite{kotliar2006electronic}. The KVQA aims to compute the impurity retarded Green's functions,

\begin{equation}\label{eq:gf}
G^R_{i\sigma}(z) = ||c^{\dagger}_{i\sigma}\ket{GS}||^2 g_{\phi^+_{i\sigma}}(z) - ||c_{i\sigma}\ket{GS}||^2 g_{\phi^-_{i\sigma}}(-z),
\end{equation}
via their continued fraction expansion,

\begin{equation}\label{eq:cfe}
g_{\phi}(z) = \bra{\phi} [z - \tilde{H}]^{-1} \ket{\phi} = \frac{1}{z - a_0 - \frac{b_1^2}{z - a_1 - \frac{b_2^2}{z - a_2 - \ldots}}},
\end{equation}
in terms of Krylov basis coefficients, $a_n$ and $b_n$, and where $z$ is a complex frequency and $\tilde{H} = H - E_{GS}$ \cite{jamet2021krylov}. Eq.~\ref{eq:gf} is taken at zero temperature, but the following arguments can be generalized to thermal distributions. The $a_n$ and $b_n$ are obtained from Lanczos iterations,

\begin{equation}\label{eq:li}
\begin{split}
b_{n}^2 &= \bra{\chi_{n-1}}\tilde{H}^2 \ket{\chi_{n-1}} - a_{n-1}^2 - b_{n-1}^2\\
\ket{\chi_n} &= \frac{1}{b_n}[(\tilde{H} - a_{n-1})\ket{\chi_{n-1}} - b_{n-1}\ket{\chi_{n-2}}]\\
a_n &= \bra{\chi_n} \tilde{H} \ket{\chi_n},
\end{split}
\end{equation}
with $b_0=0$, $\ket{\chi_0}=\ket{\phi}$, and $a_0=\bra{\chi_0} \tilde{H} \ket{\chi_0}$. Two sets of Lanczos iterations need to be performed in order to compute Eq.~\ref{eq:gf}, one for each of the initial Krylov vectors, $\ket{\phi^+_{i\sigma}} = c^{\dagger}_{i\sigma}\ket{GS}/||c^{\dagger}_{i\sigma}\ket{GS}||$ and $\ket{\phi^-_{i\sigma}} = c_{i\sigma}\ket{GS}/||c_{i\sigma}\ket{GS}||$.

First, we provide a new method to compute initial Krylov vectors. In prior work, $\ket{\phi^{\pm}_{i\sigma}}$ have been computed variationally, which is unintuitive since the application of a single creation or annihilation operator to the ground state wavefunction is a non-unitary operation \cite{jamet2021krylov, PhysRevB.107.165155}. Many quantum processors now offer the ability to measure qubits mid-circuit and reset them \cite{graham2023mid, PhysRevX.13.041034, google2023suppressing}. Meanwhile, creation and annihilation operators can be expressed as $c^{\dagger}_{i \sigma} \propto (\prod_{\nu < (i,\sigma)} Z_{\nu}) X_{i\sigma} P_{i\sigma}(0)$ and $c_{i \sigma} \propto (\prod_{\nu < (i,\sigma)} Z_{\nu}) X_{i\sigma} P_{i\sigma}(1)$, where $P_{i\sigma}(0)$ and $P_{i\sigma}(1)$ are projectors onto the impurity qubit $\ket{0}$ and $\ket{1}$ state, respectively. So, in order to prepare each initial Krylov vector, $\ket{\phi^{+}_{i\sigma}}$ ($\ket{\phi^{-}_{i\sigma}}$), one prepares the ground state of Eq.~\ref{eq:AIM} via VQE with the SPA, measures the state of the relevant impurity qubit until $\ket{0}$ ($\ket{1}$) is measured, and then applies the operator $(\prod_{\nu < (i,\sigma)} Z_{\nu}) X_{i\sigma}$. If after $M$ tries, the result $\ket{0}$ ($\ket{1}$) is not observed, it implies that the associated norm in Eq.~\ref{eq:gf} is zero with certainty $\sim 1/\sqrt{M}$, and the corresponding Lanczos iterations and continued fraction do not need to be computed. Also note that by collapsing the wavefunction of the impurity qubit via mid-circuit measurement, the initial Krylov vectors automatically become normalized. In order to compute the unrenormalized modulus, one can use the ground state wavefunction: $||c^{\dagger}_{i\sigma}\ket{GS}||^2 = (1+\bra{GS} Z_{i\sigma} \ket{GS})/2$ and $||c_{i\sigma}\ket{GS}||^2 = (1-\bra{GS} Z_{i\sigma} \ket{GS})/2$ as in \cite{jamet2021krylov}.

\begin{figure}
\centering
\includegraphics[width=\linewidth]{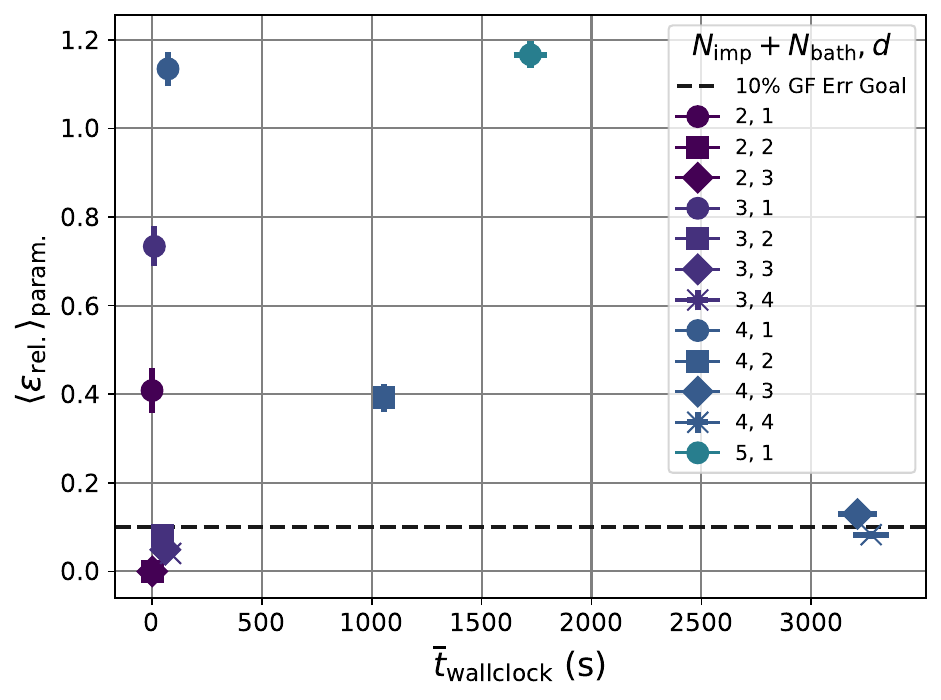}
\caption{\label{fig:gf_relative_errors} Expressivity of the symmetry-preserving ansatz for computing Lanczos vectors. $\langle \epsilon_{\text{rel.}} \rangle_{\text{param.}}$ is the relative error between the classical and variationally prepared Green's functions, averaged over Hamiltonian parameters and $\Bar{t}_{\text{wallclock}}$ is the parameter-averaged supercomputer wallclock time, in seconds, required to converge all of the Lanczos interations and compute the resulting Green's function. Error bars are one standard error of the mean over 100 parameter sets. Purple, violet, blue, and green points represent 2, 3, 4, and 5 sites, respectively, while circles, squares, diamonds, and stars represent a depth of 1, 2, 3, and 4 respectively. The black dashed line demarks the $10\%$ relative error threshold. Note that linear SPA depth is sufficient to prepare all Green's functions to within $10\%$ relative error up to 4-sites.}
\end{figure}

If $(N^{GS}, S_z^{GS})$ is the charge-spin sector in which the ground state of Eq.~\ref{eq:AIM} resides, then the two associated Krylov subspaces in which the Lanczos iterations will be performed are related via $(N^{\pm}, S_z^{\pm}) = (N^{GS} \pm 1, S_z^{GS} \pm (\delta_{\sigma \uparrow}- \delta_{\sigma \downarrow}))$. In Ref.~\cite{jamet2021krylov}, Jamet et al. observed that a sufficient condition to compute the second line in Eq.~\ref{eq:li} is to find a variational state, $\ket{\chi(\boldsymbol{\theta}_n)}$, such that $\braket{\chi(\boldsymbol{\theta}_n)|\chi(\boldsymbol{\theta}_{n-2})} = \braket{\chi(\boldsymbol{\theta}_n)|\chi(\boldsymbol{\theta}_{n-1})} = 0$ and $\braket{\chi(\boldsymbol{\theta}_n)|H|\chi(\boldsymbol{\theta}_{n-1})}=b_n$, where $\ket{\chi(\boldsymbol{\theta}_{n-2})}$ and $\ket{\chi(\boldsymbol{\theta}_{n-1})}$ are the corresponding variational states at Lanczos iterations $n-1$ and $n-2$, respectively. The variational Lanczos state can be obtained at each iteration by minimization of the cost function 
\begin{equation}\label{eq:lanczos_cost}
\begin{split}
    F(\boldsymbol{\theta}_n) = &\lambda_1(|\braket{N^{\pm}, S_z^{\pm}; \boldsymbol{\theta}_n | H | N^{\pm}, S_z^{\pm}; \boldsymbol{\theta}_{n-1}}| - |b_n|)^2\\
    &+ \lambda_2 |\braket{N^{\pm}, S_z^{\pm}; \boldsymbol{\theta}_n| N^{\pm}, S_z^{\pm}; \boldsymbol{\theta}_{n-1}}|^2 \\
    &+ \lambda_3 |\braket{N^{\pm}, S_z^{\pm}; \boldsymbol{\theta}_n| N^{\pm}, S_z^{\pm}; \boldsymbol{\theta}_{n-2}}|^2
\end{split}
\end{equation}
%
%
where all parametrized unitary evolutions are performed according to Eq.~\ref{eq:ns_trial} in the correct $(N^{\pm}, S_z^{\pm})$-constrained Krylov subspaces, and $(\lambda_1, \lambda_2, \lambda_3)$ are optional Lagrange multipliers.

We compute the Lanczos iterations both classically via Eq.~\ref{eq:li} and quantum-variationally using Eq.~\ref{eq:lanczos_cost} and use Eqs.~\ref{eq:gf}-\ref{eq:cfe} to construct exact, $G^{R, \text{exact}}_{I, \uparrow}(\omega + i \eta)$, and variational approximations to, $G^{R, \text{var.}}_{I, \uparrow}(\omega + i \eta)$, the retarded impurity Green's function. We choose the spin-up orbital without loss of generality due to the system's $\mathbb{Z}_2$ symmetry; and $\eta=0.1$ is a small line-broadening parameter for visual aid, as is typically used \cite{selisko2024dynamical}. To assess the quality of the variational Lanczos method as computed by noiseless emulation, we compute the relative error between the exact and quantum-variationally computed Green's functions
\begin{equation}\label{eq:rel_err}
    \epsilon_{\text{rel.}} = \frac{|G^{R, \text{var.}}_{I, \uparrow}(\omega + i \eta) - G^{R, \text{exact}}_{I, \uparrow}(\omega + i \eta)|}{|G^{R, \text{exact}}_{I, \uparrow}(\omega + i \eta)|}.
\end{equation}
Fig.~\ref{fig:gf_relative_errors} indicates the expressivity of our ansatz in variationally preparing the Lanczos vectors used to construct the Green's function. Specifically, the SPA at linear depth, i.e., $N_{\text{bath}}+N_{\text{imp}} = d$ is sufficient to compute all Green's functions to within a seed-averaged relative error of $10\%$ up to $N_{\text{bath}}+N_{\text{imp}}=4$, corresponding to eight qubits. The average is taken over Hamiltonian parameter seeds. Fig.~\ref{fig:combined_gfs} demonstrates that a $10\%$ relative error in the Green's function is often sufficient to capture a large majority of the quantitatively important features of the impurity response. We remark, however, that more data is needed at larger system sizes to draw better conclusions about whether linear ansatz depth is in general sufficient to prepare all the Lanczos vectors required to compute Green's functions to reasonable accuracy. Moreover, the trainability of our ansatz using Eq.~\ref{eq:lanczos_cost} as a cost function remains an open question, with the linear-depth, 4-site Green's function computation taking roughly an order of magnitude longer than the linear-depth, 3-site computation as measured by supercomputer wallclock mean time. In particular, one can see that in the 4-site case, even though the Krylov subspace dimension is being kept fixed, the wallclock time increases drastically as the depth of the ansatz is increased, indicating that trainability is degrading rather than the number of Lanczos vectors being the cause of prolonged runtime. However, given that the SPA is efficiently trainable for ground state preparation, it is unclear if the form of the cost function in Eq.~\ref{eq:lanczos_cost}, or the ansatz itself, or the interplay between the two, is the cause of degraded trainability. We include the 5-site, $d=1$ point for reference, but note that the simulation wallclock time becomes untenable for the $(5, 2)$ and larger cases. Hence expanding our results in the future would benefit from more sophisticated circuit simulation techniques like tensor network methods run on GPUs. These methods, however, will not help with the combinatorial proliferation of Lanczos vectors in the Krylov subspaces of larger system size.

In Fig.~\ref{fig:combined_gfs} we examine the relationship between the qualitative and quantitative performance of variational Green's function preparation for a selection of Hamiltonian seeds at $N_{\text{imp}}+N_{\text{bath}}=d=4$. Seeds 18 (Fig.~\ref{fig:combined_gfs}a) and 84 (Fig.~\ref{fig:combined_gfs}b) both exhibit both very good quantitative and qualitative accuracy, with a relative error of $\sim .68\%$ and $\sim 5.7\%$, respectively. The relative error in seed 84 arises from small amplitude fluctuations near the Fermi level. Seed 24 (Fig.\ref{fig:combined_gfs}c) exhibits a relative error of $\sim 30\%$ arising mainly from small shifts in peak positions and bifurcation in the main peak near $\sim 10$ eV. Finally, seed 62 (Fig.~\ref{fig:combined_gfs}) represents an example of a variationally prepared Green's function that, despite its large relative error of $\sim 84\%$ due to small amplitude fluctuations, peak shifts, and peak bifurcations, retains a good qualitative description of the response function.

Figs.~\ref{fig:gf_relative_errors}-\ref{fig:combined_gfs} suggest that at even larger system sizes, the SPA may continue to be a useful tool to prepare the many-body states necessary to approximate one-particle Green's functions in the frequency domain, either via the Lanczos method or in other representations like the Lehmann representation. However, using a quantum processor to prepare many-body states to compute Green's functions in the frequency domain is ultimately less efficient than using a quantum processor to compute Green's functions directly in the time domain. There is a rigorous exponential separation between quantum and classical computers' ability to perform time evolution, while there often at best only a polynomial separation for quantum state preparation \cite{lee2023evaluating}. In the next section we propose a new algorithm for computing arbitrary $m$-point time correlation functions on a quantum computer.

\begin{figure*}
\centering
\includegraphics[width=.85\linewidth]{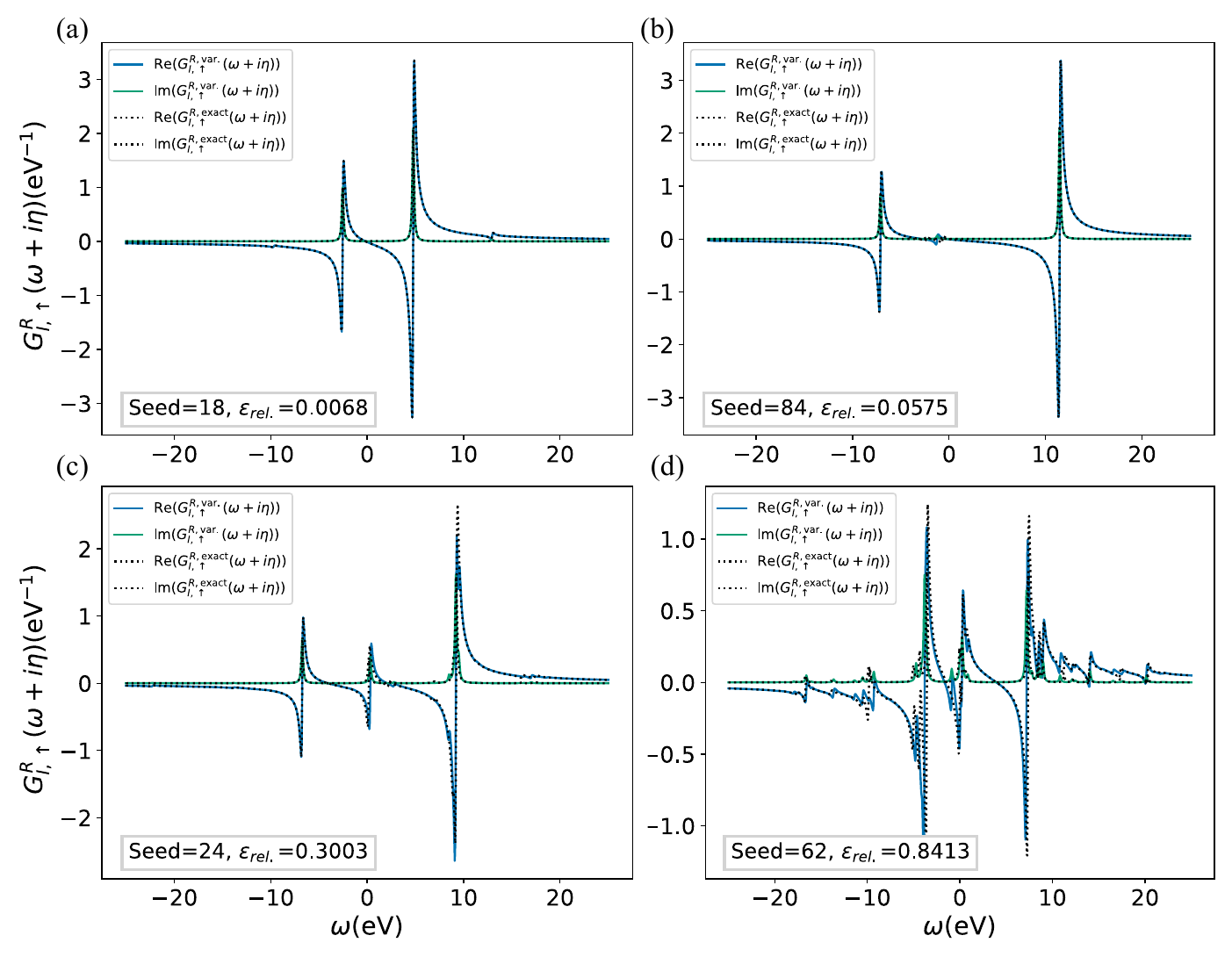}
\caption{\label{fig:combined_gfs} Quantitative and qualitative accuracy of variationally prepared Green's functions. All subplots are for $N_{\text{imp}}+N_{\text{bath}}=d=4$. Retarded Green's functions are computed via the continued fraction representation. Solid lines refer to Lanczos vectors and coefficients computed variationally using the symmetry-preserving ansatz. Dashed lines refer to Lanczos vectors and coefficients computed exactly. (a) Hamiltonian parameter seed 18 has $\sim .68\%$ relative error and good very good qualitative agreement between the exact and variationally prepared Green's function. (b) Seed 84 has $\sim 5.7\%$ relative error and exhibits only small-amplitude fluctuations near the Fermi level. (c) Seed 24 has a $\sim 30\%$ relative error and  contains both small shifts in peak position as well as bifurcations in large amplitude peaks. (d) Seed 62 has a relative error of $\sim 84\%$ due to amplitude fluctuations, peak shifts, and peak bifurcations, and yet qualitatively reproduces many of the main features of the spectrum.}
\end{figure*}

\section{\label{sec:correlators} Computation of Many-Body Correlation Functions}

While the computation of the single-particle Green's function is sufficient in DMFT to enforce self-consistency with the weakly-correlated bath theory, higher-order correlation functions are often needed to accurately model properties like magnetic susceptibilities and electron-hole propagators in realistic simulations of strongly-correlated materials \cite{martin2016interacting, acharya2023theory}. Prior deterministic theorems and algorithms have been provided to perform this computation on a quantum computer (see, for example, \cite{pedernales2014efficient, terhal2000problem}). Here we propose another such algorithm and provide an explicit circuit construction, which combines unitary time evolution, mid-circuit measurement-conditioned operations, and a modified Hadamard test.

Consider a collection of real-time fermionic operators, $\{f_{\alpha_j}(t_j) \: |\: j=1\ldots m\}$ in the Heisenberg picture
\begin{equation}\label{eq:heisenberg_fermion}
    f_{\alpha_j}(t) = e^{iHt} f_{\alpha_j} e^{-iHt},
\end{equation}
where each $f_{\alpha_j}$ stands either for a creation ($c_{\alpha_j}^{\dagger}$) or annihilation ($c_{\alpha_j}$) operator on the spin-orbital $\alpha_j$ in the Schrodinger picture. A generic $m$-point correlation function primitive, from which more complex objects like time-ordered correlators can be constructed, is
\begin{equation}\label{eq:general_correlator}
\begin{split}
G(m, \ldots, 1) &= \langle f_{\alpha_m}(t_m) f_{\alpha_{m-1}}(t_{m-1})\ldots f_{\alpha_1}(t_{\alpha_1}) \rangle \\
&= \bra{GS} e^{iHt_m} f_{\alpha_m} e^{-iH(t_m-t_{m-1})} \ldots \\
& \: \: \: \: \: \indent \indent \:  \ldots e^{-iH(t_2-t_1)} f_{\alpha_1} e^{-iHt_1} \ket{GS}.
\end{split}
\end{equation}
We compute correlations in the ground state, leaving generalizations to thermal states to future work. The problem with computing generic $m$-point correlators is that the operator string in Eq.~\ref{eq:general_correlator} is not necessarily unitary, precluding any straightforward application of the Hadamard test. Our approach is therefore to compute Eq.~\ref{eq:general_correlator} by modifying the Hadamard test to incorporate controlled operators that can be constructed from common operations available on a quantum computer, which is to say \textit{both} unitary operators \textit{and} projective measurements.


The first step in this modification is to make the following definition,
\begin{equation}\label{eq:normed_fermion}
\Tilde{f}_{\alpha} \equiv \prod_{\nu < \alpha} Z_{\nu} X_{\alpha} P_{\alpha}(z),
\end{equation}
which was previously used in Sec.~\ref{sec:spkvqa} to compute initial Krylov vectors in the Jordan-Wigner representation. Here again, $Z_{\nu}$ and $X_{\alpha}$ are Pauli operators and $P_{\alpha}$ is a projector onto either the $\ket{z} = \ket{0}$ or $\ket{z}=\ket{1}$ state of orbital $\alpha$ depending on if $f_{\alpha}$ refers to a creation or annihilation operator, respectively. The action of Eq.~\ref{eq:normed_fermion} on a generic state $\ket{\psi}$ is
\begin{equation}\label{eq:nermion_action}
    \tilde{f}_{\alpha}\ket{\psi} = 
    \begin{cases}
        0 & \text{if \: $||f_{\alpha}\ket{\psi}|| = 0$}\\
        f_{\alpha}\ket{\psi}/||f_{\alpha}\ket{\psi}|| & \text{else},
    \end{cases}
\end{equation}
which can be re-written compactly as
\begin{equation}\label{eq:fermion_action}
    f_{\alpha}\ket{\psi} = ||f_{\alpha}\ket{\psi}|| \times \tilde{f}_{\alpha}\ket{\psi}.
\end{equation}
We use Eq.~\ref{eq:fermion_action} to re-write the operator string Eq.~\ref{eq:general_correlator} in terms of the renormalized operators, $\{\tilde{f}_{\alpha_j}\}$.

\begin{theorem}\label{thm:correlator_theorem} 
    Take $t_0=0$ and $f_{\alpha_0}=1$ and let
    $\ket{\psi_j} \equiv \prod_{k=j}^1 e^{-iH(t_k-t_{k-1})}f_{\alpha_{k-1}}\ket{GS}$ be an un-normalized statevector. The $m$-point correlation function, Eq.~\ref{eq:general_correlator}, can be computed as
       \begin{equation}\label{eq:computable_correlator}
        G(m, \ldots, 1) = \tilde{G}(m, \ldots, 1) \prod_{j=m}^1 ||f_{\alpha_j}\ket{\psi_j}||,
    \end{equation}
    where all $m$ factors $||f_{\alpha_j}\ket{\psi_j}||$ can be computed recursively using $m$ quantum circuits to measure expectation values of simple local observables of the form $f^{\dagger}_{\alpha_j}f_{\alpha_j}$, and  if all such factors are non-zero, then the quantity $\tilde{G}(m,\ldots,1)=\langle \prod_{j=m}^1 \tilde{f}_{\alpha_j}(t_j) \rangle$ is computable by a straightforward application of the Hadamard test, since in this instance the operator string, $\prod_{j=m}^1 \tilde{f}_{\alpha_j}(t_j)$, is physically implementable.
\end{theorem}

\begin{proof}
    First, we show that Eq.~\ref{eq:general_correlator} can be re-written as Eq.~\ref{eq:computable_correlator}. We show this recursively via construction, beginning with Eq.~\ref{eq:general_correlator} expressed as
    \begin{equation}
        G^{(m)} = \bra{GS} e^{iHt_m} f_{\alpha_m} \ket{\psi_m},
    \end{equation}
    where we have assigned the shorthand $G(m, \ldots, 1)\rightarrow G^{(m)}$. Using Eq.~\ref{eq:nermion_action}, consider the action of the last fermionic operator on $\ket{\psi_m}$ as the base case, and noting that norms commute through operators:
    \begin{equation}
    \begin{split}
        e^{iHt_m} &f_{\alpha_m}\ket{\psi_m} = ||f_{\alpha_m}\ket{\psi_m}|| \times e^{iHt_m} \tilde{f}_{\alpha_m} \ket{\psi_m}\\
        &= ||f_{\alpha_m}\ket{\psi_m}|| \times \tilde{f}_{\alpha_m}(t_m) e^{iHt_{m-1}} f_{\alpha_{m-1}} \ket{\psi_{m-1}}
    \end{split}
    \end{equation}
    Now, assume that after $m-l+1$ replacements we have
    \begin{equation}
    \begin{split}
        G^{(m)} = \bra{GS} \prod_{j=m}^{m-l+1} &||f_{\alpha_j}\ket{\psi_j}|| \tilde{f}_{\alpha_j}(t_j) \times \\
        &\times e^{iHt_{m-l}} f_{\alpha_{m-l}} \ket{\psi_{m-l}}\\
        = \bra{GS} \prod_{j=m}^{m-l} &||f_{\alpha_j}\ket{\psi_j}|| \tilde{f}_{\alpha_j}(t_j) \times \\
        &\times e^{iHt_{m-l-1}} f_{\alpha_{m-l-1}} \ket{\psi_{m-l-1}},
    \end{split}
    \end{equation}
    which can be carried all the way to the $j=1$ case, resulting in Eq.~\ref{eq:computable_correlator}.
    
    Next, we prescribe how to recursively compute the norms in Eq.~\ref{eq:computable_correlator}. Consider, first, the $j=1$ base case
    \begin{equation}\label{eq:norm_base_case}
        ||f_{\alpha_1}\ket{\psi_1}|| = ||f_{\alpha_1}e^{-iHt_1}\ket{GS}||
    \end{equation}
    Eq.~\ref{eq:norm_base_case} is straightforwardly-computed as an expectation value of the local operator $f_{\alpha_1}^{\dagger}f_{\alpha_1}$ sampled from the output of the quantum circuit $e^{-iHt_1}\ket{GS}$. Sampling the circuit output $M$ times will give the expectation value with precision $\sim 1/\sqrt{M}$ from standard shot statistics. For illustration purposes consider the $j=2$ case:
    \begin{equation}\label{eq:norm_second_case}
    \begin{split}
        ||f_{\alpha_2}\ket{\psi_2}||&=||f_{\alpha_2}e^{-iH(t_2-t_1)}f_{\alpha_1}\ket{\psi_1}||\\
        &= ||f_{\alpha_2}e^{-iH(t_2-t_1)}\tilde{f}_{\alpha_1}\ket{\psi_1}||\times ||f_{\alpha_1}\ket{\psi_1}||.
    \end{split}
    \end{equation}
    Hence, we can compute Eq.~\ref{eq:norm_second_case} with the previously-computed result of Eq.~\ref{eq:norm_base_case} and by sampling the output of the new circuit, $e^{-iH(t_2-t_1)}\tilde{f}_{\alpha_1}e^{-iHt_1}\ket{GS}$, to compute the expecation value of the new local observable, $f_{\alpha_2}^{\dagger}f_{\alpha_2}$. This methodology generalizes so that if one has already computed the factors $\{||f_{\alpha_j}\ket{\psi_j}|| \: \text{for} \: j=1, \ldots, l-1\}$, then the $l^{\text{th}}$ factor can be computed as
    \begin{equation}\label{eq:norm_recursion}
    \begin{split}
        ||f_{\alpha_l}\ket{\psi_l}|| = ||f_{\alpha_l}&\prod_{k=l-1}^0 e^{-iH(t_{k+1}-t_k)}\tilde{f}_{\alpha_k}\ket{GS}||\\
        \times &\prod_{j=l-1}^1||f_{\alpha_j}\ket{\psi_j}||.
    \end{split}
    \end{equation}
    Altogether, evaluation of Eq.~\ref{eq:computable_correlator} requires sampling the output of $m$ quantum circuits a total of $M$ times each, as well as implementing the Hadamard test twice to compute the real and imaginary parts of $\tilde{G}^{(m)}$.
\end{proof}

From a practical perspective one uses Thm.~\ref{thm:correlator_theorem} in the following manner. First, the norms in Eq.~\ref{eq:computable_correlator} are computed recursively as in Eqs.~\ref{eq:norm_base_case}-\ref{eq:norm_recursion}. If at any point, one of the norms vanishes, the recursion can be stopped and the correlation function is zero. Only if the product of norms is computed to be non-zero to within the desired precision of the calculation should the Hadamard test be performed to compute $\tilde{G}^{(m)}$, since it is only in this instance that the renormalized Heisenberg operator string can be considered as non-vanishing and physically implementable using standard quantum circuit instructions. For concreteness we discuss a specific example below.

\subsection{Single-Particle Green's Functions}

The greater and lesser Green's functions are defined, respectively as \cite{martin2016interacting}
\begin{equation}\label{eq:greater_gf}
    G^>(2,1) = -i \langle c_{\alpha_2}(t_2) c_{\alpha_1}^{\dagger}(t_1) \rangle
\end{equation}
and
\begin{equation}
        G^<(2,1) = i \langle c_{\alpha_1}^{\dagger}(t_1) c_{\alpha_2}(t_2) \rangle,
\end{equation}
from which other Green's functions can be constructed. For instance, using the Heavyside step function, the retarded Green's function is
\begin{equation}
    G^R(2, 1) = \Theta(t_2 - t_1)(G^>(2,1) - G^<(2,1)).
\end{equation}
Taking Eq.~\ref{eq:greater_gf} as an example, defining the standard time-translationally invariant quantity $t \equiv t_2-t_1$, and computing the expectation value in the ground state via Thm.~\ref{thm:correlator_theorem}, we find
\begin{equation}\label{eq:computable_greater_gf}
    G^>(2,1) = -i \bra{GS} \tilde{c}_{\alpha_2}(t) \tilde{c}^{\dagger}_{\alpha_1}\ket{GS} ||c_{\alpha_2}\ket{\psi_2}|| \: ||c^{\dagger}_{\alpha_1}\ket{GS}||
\end{equation}
where $||c^{\dagger}_{\alpha_1}\ket{GS}||=\sqrt{(1+\bra{GS}Z_{\alpha_1}\ket{GS})/2}$ is computed simply by reading out the ground state wavefuction in the computational $z$-basis, as previously noted (Fig.~\ref{fig:correlator_circuits}a), and
\begin{equation}
||c_{\alpha_2}\ket{\psi_2}||=||c^{\dagger}_{\alpha_1}\ket{GS}||\sqrt{(1-\bra{\tilde{\psi}_2}Z_{\alpha_2}\ket{\tilde{\psi}_2})/2},
\end{equation}
with the state $\ket{\tilde{\psi}_2}=e^{-iHt}\tilde{c}^{\dagger}_{\alpha_1}\ket{GS}$ prepared from Eq.~\ref{eq:normed_fermion} applied to the ground state followed by Hamiltonian time evolution for a time $t$ (Fig.~\ref{fig:correlator_circuits}b). Finally, if both norms are non-zero, the real and imaginary components of the propagator in Eq.~\ref{eq:computable_greater_gf} are computed by two applications of the non-unitary Hadamard test (Fig.~\ref{fig:correlator_circuits}c).

\begin{figure*}
\centering
\includegraphics[width=\linewidth]{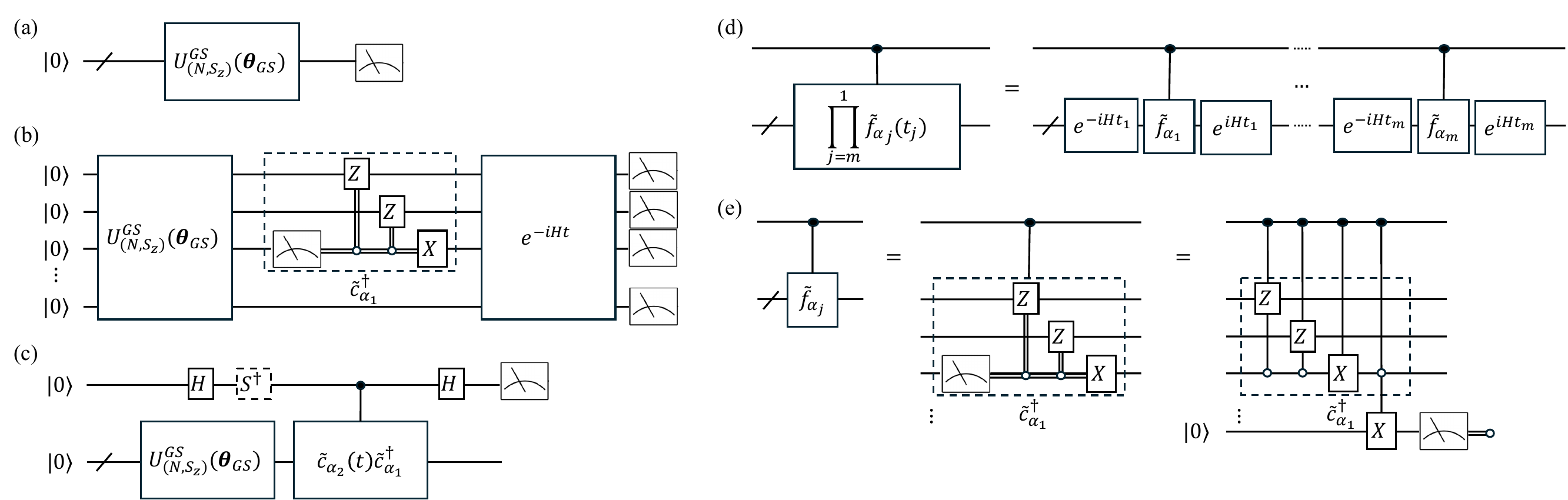}
\caption{\label{fig:correlator_circuits} Circuits to compute $m$-point correlation functions of fermionic Hamiltonians. (a) Ground state preparation and readout in the $z$-basis to compute the expectation value $||c_{\alpha_1}^{\dagger}\ket{GS}||$. (b) Preparation of the state $\ket{\tilde{\psi}_2}$ and readout in the $z$-basis to compute $||c_{\alpha_2}\ket{\psi_2}||$. This circuit can be generalized to compute an arbitrary norm $||f_{\alpha_j}\ket{\psi_j}||$. (c) The modified Hadamard test when the non-unitary Heisenberg operator string is the specific instance $\prod_{j=m}^1 \tilde{f}_{\alpha_j}(t_j) = \tilde{c}_{\alpha_2}(t)\tilde{c}^{\dagger}_{\alpha_1}$. The optional implementation of the $S^{\dagger}$ gate dictates if the real or imaginary part is computed. (d) Any controlled non-unitary Heisenberg fermion string can be implemented by a sequence of unitary time evolutions and single controlled Schr\"odinger fermion gates. (e) Example implementation of a single controlled Schr\"odinger fermion. In this instance a creation operator, $\tilde{c}_{\alpha_1}^{\dagger}=Z_1\otimes Z_2 \otimes X_3 P_3(0)$, needs to be applied to the spin-orbital $\alpha_1$, which has been assigned to the third qubit, conditioned on if the Hadamard test ancilla is in the $\ket{1}$ state. Since the gate must be coherently controlled, we push the projective measurement past the classically-controlled $Z$ gates using the principle of deferred measurement and commute it past the $X$ gate. Then, the entire sequence of controlled operations can be represented by a string of Toffoli-equivalent gates, a CNOT, and a controlled-projector, which uses an additional ancilla.
}
\end{figure*}

Fig~\ref{fig:correlator_circuits}d shows the decomposition of an arbitrary controlled Heisenberg fermion string into a sequence of time evolutions and controlled Schr\"odinger fermion operators. Note that if the control qubit is in the $\ket{0}$ state, none of the fermion operators are applied and all the forward and backward time-evolution operators multiply to the identity. Using the principle of deferred measurement and the relationship $X_{\alpha_j} P_{\alpha_j}(0)=P_{\alpha_j}(1)X_{\alpha_j}$ applied to $\tilde{c}^{\dagger}_{\alpha_1}$ as an example, we decompose each fermion operator into $\alpha_j$ Toffoli-equivalent gates (up to single-qubit rotations), one CNOT, and a single ancilla, post-selected on the $\ket{0}$ state used to implement a controlled projection operator, C$P_{\alpha_j}(1)$ (Fig.~\ref{fig:correlator_circuits}e). The ancilla can be directly reused in the implementation of subsequent fermionic operators. To minimize the number of Toffolis required to compute $m$-point impurity correlation functions, one should therefore index the impurity as low as possible and if there are multiple impurity sites, they should be grouped together by index. In future work, Thm.~\ref{thm:correlator_theorem} can be applied directly, along with the various circuit decompositions in this section, to the computation of quantities of typical interest in DMFT, such as impurity electron-hole propagators and magnetic susceptibilities \cite{PhysRevB.86.125114, martin2016interacting, PhysRevB.96.035114, PhysRevB.106.085124}.

\section{\label{sec:conclusion} Conclusion}

Herein, we have provided two explicit circuit constructions for efficiently computing the central quantities in dynamical mean-field theory. The first construction is a hardware-adaptable and symmetry-preserving variational ansatz that can be used to prepare many-body states of the Anderson impurity model. We show using numerical emulation that this ansatz can prepare ground states of the single-impurity Anderson model in depth roughly linear in the number of bath orbitals and in sub-quartic training time for a set of small to moderate-size models, indicating both good expressibility near the Lieb-Robinson bound and efficient trainability. Moreover, we show that the ansatz can be used to prepare other many-body states like the Lanczos vectors used to compute single-particle Green's functions in the continued fraction representation, although the scalability of this method remains in question, both from the perspective of expressibility and trainability and due to the fact that an exponential number of Lanczos iterations needs to be performed to compute exact Green's functions in the worst case. As such, the second circuit construction is a new method for computing arbitrary $m$-point correlation functions of fermionic systems that uses a combination of time-evolution, mid-circuit measurement, and a modified version of the Hadamard test. Aside from state preparation, which our first construction addresses, and the costs involved with time-evolving under the Anderson impurity model Hamiltonian, the main cost of this method is in the Toffoli complexity, which in the Jordan-Wigner representation, is linear in both $m$ and $N_{\text{imp}}$ given that the impurity orbitals can be indexed such that they are grouped together and assigned low indices.

Our results support the notion that electronic structure computations of strongly correlated materials that use quantum processors as impurity solvers alongside state-of-the-art classical mean-field theories constitute a promising path forward to practical quantum advantage. This notion is based on three observations made here and elsewhere \cite{jamet2023anderson}. First, state-of-the-art mean-field theories like GW theory are often sufficient to capture a large majority of the weakly-correlated physics and chemistry in describing real materials \cite{acharya2023theory}. Second, ground-state preparation of the Anderson impurity model with $O(1)$ impurity orbitals is quasipolynomially efficient both classically and quantumly \cite{bravyi2017complexity}. And third, computation of time-dependent impurity response functions, like one- and two-body Green's functions, involves Hamiltonian time evolution, for which there is a known exponential runtime classically \cite{haah2021quantum}. Meanwhile, we show by construction, that $m$-point correlation functions can be computed efficiently on a quantum computer given access to mid-circuit measurement, corroborating theoretical results suggesting this to be true \cite{pedernales2014efficient, terhal2000problem}. These observations indicate that classical computers are best suited to solving the weakly-correlated bath theory, validating and helping to prepare the ground state of the impurity model, and enforcing DMFT self-consistency, while the quantum computer is best suited to helping to solve for and encoding the ground state and deterministically computing correlation functions in the encoded ground state. Interestingly, such a workflow obviates the need for typical quantum chemistry subroutines like quantum phase estimation (although it may be used to increase the overlap with the true ground state). Moreover, the fidelity of the overall computation can be systematically improved as the fidelity of quantum computers improves, by enlarging the number of impurity orbitals to account for more spatial correlation or by using more orbitals to present a more fine-grained description the bath.

Within this computational paradigm, there are a number of open questions and potential avenues for future research. First and foremost is to assess the expressibility and trainability of our ansatz at both greater system sizes, $N_{\text{bath}}> 6$, and with more impurity sites $N_{\text{imp}}>1$, especially when simulated and trained on a classical computer using fermionic tensor network compression techniques \cite{provazza2024fast} and in the presence of noise. Along these lines, it would be interesting to investigate augmentations of our ansatz that include cooling mechanisms to remove errors via auxillary qubit reset as have been proposed and demonstrated  recently in a number of spin models \cite{matthies2022programmable, mi2024stable, lloyd2024quasiparticle}. Such augmentations may prove useful in preparing both ground and thermal states of the Anderson impurity model, which could have applicability in the quantum computation of high-temperature superconducting materials. In addition, recent results in dynamic circuits for preparing matrix product and higher-dimension tensor network states have shown promise for state preparation in constant and sub-Lieb-Robinson depth \cite{sahay2024finite, smith2024constant}. Evaluating the applicability of these methods to AIM state preparation is likely to be fruitful. Finally, the theoretical method devised herein to compute $m$-point correlation functions should be experimentally validated and used to perform resource estimations on when leveraging quantum computers to compute two-particle observables might confer quantum advantage in describing the physics of strongly correlated materials.

\begin{acknowledgments} 
E. B. J. thanks Fran\c{c}ois Jamet and Ivan Rungger for discussions and material on the Krylov variational quantum algorithm; Swagata Acharya and Mark van Schilfgaarde for discussions on the foundations of DMFT and the properties of the Anderson impurity model; Ben Hall and David Owusu-Antwi for discussions on symmetry-preserving ansatzae; and Rich Rines for discussions on circuit decompositions. This work was authored in part by the National Renewable Energy Laboratory, managed and operated by the Alliance for Sustainable Energy, LLC for the U.S. Department of Energy (DOE) under Contract No. DE-AC36-08G028308. A portion of this research was performed using computational resources sponsored by the
U.S. Department of Energy’s Office of Energy Efficiency and Renewable Energy and located at the National Renewable Energy Laboratory. The views expressed in the article do not necessarily represent the views of the DOE or the U.S. Government. The U.S. Government retains and the publisher, by accepting the article for publication, acknowledges that the U.S. Government retains a nonexclusive, paid-up, irrevocable, worldwide license to publish or reproduce the published form of this work, or allow others to do so. for U.S. Government purposes.
\end{acknowledgments}

\bibliography{bibliography}

\end{document}